 \RequirePackage[2018-12-01]{latexrelease}
\documentclass{rQUF2e}

\usepackage{epstopdf}
\usepackage{subfigure}

\usepackage{url}
\usepackage{booktabs}
\usepackage{tabulary}
\usepackage{colortbl}
\usepackage{multicol}
\usepackage{multirow}
\usepackage{diagbox}
\def\*#1{\mathbf{#1}}
\def\~#1{\boldsymbol{#1}}
\usepackage{natbib}

\theoremstyle{plain}
\newtheorem{theorem}{Theorem}[section]

\newtheorem{proposition}[theorem]{Proposition}

\theoremstyle{definition}
\newtheorem{definition}{Definition}

\theoremstyle{remark}
\newtheorem{remark}{Remark}
\newtheorem{example}[theorem]{Example}

\begin{document}


\title{Explaining Risks: Axiomatic Risk Attributions for Financial Models}

\author{Dangxing Chen$^*$\thanks{$^\ast$Corresponding author.
Email: dangxing.chen@gmail.com} \\
\affil{Duke Kunshan University, China} 
}

\maketitle

\begin{abstract}
In recent years, machine learning models have achieved great success at the expense of highly complex black-box structures. By using axiomatic attribution methods, we can fairly allocate the contributions of each feature, thus allowing us to interpret the model predictions. In high-risk sectors such as finance, risk is just as important as mean predictions.  Throughout this work, we address the following risk attribution problem: how to fairly allocate the risk given a model with data?  We demonstrate with analysis and empirical examples that risk can be well allocated by extending the Shapley value framework.
\end{abstract}


\begin{classcode}C58, C71\end{classcode}

\section{Introduction}

While machine learning (ML) models have significantly improved accuracy over traditional models, their black-box nature makes them difficult to interpret. Over the past few years, extensive research has been conducted on the interpretation. Particularly, axiomatic approaches have been highly successful in solving attribution problems  \citep{lundstrom2023four,lundstrom2023unifying,lundstrom2022rigorous,lundberg2017unified,sundararajan2017axiomatic}. These methods fairly allocate contributions to features by preserving desired axioms and therefore allow researchers to interpret the black-box structure of ML models.

While there has been great success in the interpretation of models, past research has concentrated on the baseline attribution methods (BAMs) for mean predictions. In other words, given a specific prediction, we determine the contribution of each feature. When it comes to finance, such a high-stakes sector, mean predictions are not the only matter to be concerned with; risk is also crucial. Lack of proper risk management could have catastrophic consequences, as was demonstrated by the recent collapse of Silicon Valley Bank. As a result, risks are considered along with returns in many financial studies, such as modern portfolio theory \citep{ding2022drawdown,xidonas2020robust} and factor models \citep{ang2020using,bai2002determining,bai2003inferential}. 

In this work, we address the following \textbf{Risk Attribution Methods (RAMs)}: how can risk be fairly distributed given a model based on data? As a simple example, consider a portfolio $Y_t = X_{1,t} + X_{2,t}$ as the daily return of the portfolio with two assets. At each date, $Y_t$ provides only a snapshot of the portfolio's performance and does not reflect the portfolio's vulnerability to downside risks. While the existing BAMs could be applied to $Y_t$ each day to provide an explanation of how each asset contributes to the calculation as $X_{i,t}$, the randomness of $X_{i,t}$ over time is not considered and therefore the risk is neglected. The bank wishes to determine the risk by investigating the potential loss under different scenarios, especially during critical times such as financial crises. In practice, as an example, the bank might look at the portfolio's past year performance or simulations to measure the risk. The portfolio may need to be adjusted if the risk is too high. As a result, the focus will be on the risk.

This paper extends Shapley values to solve RAMs by carefully defining characteristic functions and investigating theoretical properties. We analyze and provide numerous analytical and empirical examples to demonstrate that by using the Shapley value, it is possible to fairly allocate risk and provide useful interpretations. Our main contribution can be summarized as follows. 
\begin{enumerate}
    \item We propose the framework of risk attribution problems for ML models. 
    \item We extend the Shapley value to RAMs and compare it with BAMs. 
\end{enumerate}

\paragraph{Related Literature}
In recent years, there have been extensive studies on the axiomatic approach to attribution problems in the ML community \citep{lundstrom2023unifying,lundstrom2022rigorous,lundberg2017unified,sundararajan2017axiomatic}. Although these studies have provided much insight into the properties of attribution methods for mean predictions, they are not necessarily applicable to RAMs. RAMs differ from BAMs primarily due to the presence of risk measures.
In another direction, finance researchers began to investigate RAMs \citep{tarashev2016risk,shalit2021shapley,hagan2023portfolio}. However, these studies have been restricted to specific financial models (e.g., bank's systematic risk and portfolio theory) and sometimes to specific risk measures. \textbf{To the best of our knowledge, this is the first work to study RAMs for general problems.}

The remainder of the paper is organized as follows. The Shapley values and axioms for BAMs are reviewed in Section 2. We introduce RAMs and discuss how to extend the Shapley value for RAMs in Section 3. We compare RAMs with BAMs in Section 4 in terms of axioms, highlighting what we expect or not expect in the context of RAMs. In Section 5, we provide empirical examples. Section 6 concludes.

\section{Baseline attribution methods and Shapley value}
\label{prerequisites}

We are interested in the problem of the form,
\begin{align} \label{eq:framework}
    Y = f(\*X) + \epsilon,
\end{align}
where $\*X$ is a $m$-dimensional random variable, $\epsilon$ is independent from $\*X$, and ML models are applied to approximate $f$. For simplicity, we assume $\mathcal{F}$ to be the set of real analytic functions and $f \in \mathcal{F}$. 


\subsection{Baseline attribution methods} \label{sec:BAM}

Baseline attribution methods (BAMs) involve assigning a model's prediction score to a specific input. Following \cite{lundstrom2023four}, throughout the paper, $\*x$ represents a general function input, $\overline{\*x}$ represents the input to explain, which we call the explicand, and $\*x'$ denotes a reference baseline. Without loss of generality (WLOG), we assume $\overline{\*x}>\*x'$ and $\*x'=\*0$. Then we have the following definition.

\begin{definition}[\textbf{Baseline Attribution Method (BAM)}] \label{def:BAM}
    Given an explicand $\overline{\*x}$, baseline $\*x'$, and a function $f$, a baseline attribution problem is any function of the vector form $\~{\mathcal{A}} (\overline{\*x},\*x',f) \in \mathbb{R}^m$. We use $\mathcal{A}_i(\overline{\*x},\*x',f)$ and sometimes $\mathcal{A}_i(f)$ or $\mathcal{A}_i$ to denote the $i$-th attribution of $\~{\mathcal{A}} (\overline{\*x},\*x',f)$.
\end{definition}

The following fundamental axioms are usually required for BAMs, as discussed in  \cite{sundararajan2020many}. As they apply to general problems, we call them fundamental axioms.

\begin{definition}\label{def:basic_axiom_BAM}
\textbf{Fundamental axioms for BAMs}.

$\bullet$ \textbf{Completeness (BAM): }
    $\sum_{i=1}^m \mathcal{A}_i = f(\overline{\*x}) - f(\*x').$
    
$\bullet$ \textbf{Linearity (BAM): }
    For $\alpha, \beta \in \mathbb{R}$, functions $f,g$, $\mathcal{A}_i(\alpha f + \beta g) = \alpha \mathcal{A}_i(f) + \beta \mathcal{A}_i (g)$.
    
$\bullet$ \textbf{Dummy (BAM):}
    If $\frac{\partial f(\*x)}{\partial x_i} = 0$, then $\mathcal{A}_i = 0$.
    
$\bullet$ \textbf{Symmetry (BAM):}
    Define $\*x^*$ by swapping $x_i$ and $x_j$. Suppose $f(\*x) = f(\*x^*)$, $\forall \*x$. Then if $\overline{x}_i = \overline{x}_j$ and $x_i' = x_j'$, then $\mathcal{A}_i = \mathcal{A}_j$.
\end{definition}

\subsection{Shapley values} \label{sec:Shap}


The Shapley value, introduced by \cite{shapley1953value}, takes as input a characteristic function $v:2^M \rightarrow \mathbb{R}$ with $M = \{1, \dots, m\}$ and $v(\emptyset) = 0$. The Shapley value produces attributions $\text{SH}_i(v)$ for each player $i \in M$,
\begin{align} \label{eq:Shapley}
    \text{SH}_i(v) = \sum_{S \subseteq M \backslash i} \frac{|S|! (|M|-|S|-1)!}{|M|!} (v(S \cup i) - v(S)).
\end{align}
Shapley values are uniquely determined by following fundamental axioms.  

\begin{definition}
We call the following axioms the \textbf{fundamental axioms for Shapley value}. 

 $\bullet$ \textbf{Completeness (SH):}
        $\sum_{i=1}^m \text{SH}_i(v) = v(M) - v(\emptyset)$.

$\bullet$ \textbf{Linearity (SH):}
For $\alpha, \beta \in \mathbb{R}$ and characteristic functions $v$ and $w$, 
$
     \text{SH}_i(\alpha v + \beta w) = \alpha \text{SH}_i(v) + \beta \text{SH}_i (w).
$

$\bullet$ \textbf{Dummy (SH):}
If 
$    
v(S \cup i) - v(S) = 0, \forall S \subseteq M \backslash i,
$
 then $\text{SH}_i(v) = 0$.

$\bullet$ \textbf{Symmetry (SH):}
If 
$
     v(S \cup i) = v(S \cup j), \forall S \subseteq M \backslash \{i,j\},
 $    
 then $\text{SH}_i(v) = \text{SH}_j(v)$.
 \end{definition}


For BAMs, we focus on the Baseline Shapley value (BShap), introduced by \cite{sundararajan2020many}, which calculates
\begin{align*}
    v_B(\overline{\*x},\*x',f;S) = f(\overline{\*x}_S; \*x'_{M \backslash S}).
\end{align*}
That is, baseline values replace the feature's absence.  For example, suppose $f(x_1,x_2) = x_1+x_2$, $\overline{\*x} = (\overline{x}_1, \overline{x}_2)$, $\*x' = (0,0)$, and $S=\{1\}$, then we have $v_B(\overline{\*x},\*x',f;S) = f(\overline{x}_1,0)$. We use $v_B(\overline{\*x},\*x',f;S)$ to denote the characteristic function for BAMs and denote the method as BShap (BAM). We often use $v_B(S)$ in short form.  We denote the $i$-th attribution of BShap (BAM) by $\text{BS}_i(v_B)$. We focus on the BShap since it has better theoretical properties than other variants of Shapley value in terms of preserving axioms, as discussed in \cite{sundararajan2020many}. As a result, we are more confident about using it for sectors with high stakes. Furthermore, it is verified that BShap (BAM) preserves Completeness (BAM), Linearity (BAM), Dummy (BAM), and Symmetry (BAM).

\paragraph{Monotonic axioms} Fundamental axioms focus on general problems without any prior knowledge. Finance, however, has benefited from many domain knowledge, and monotonicity is one of the most important. Examples include option pricing \citep{shreve2004stochastic,dugas2009incorporating} and credit scoring \citep{chen2023address,chen2022monotonic,repetto2022multicriteria}. As a result, many studies have examined monotonicity axioms in the content of BAMs \citep{chen2024attribution,friedman1999three,lundstrom2023four}. We discuss how monotonicity is reflected in Shapley value.

\begin{definition}\label{def:mono_axiom_Sh}
\textbf{Monotonic axioms for Shapley value}. 
\begin{itemize}
\item \textbf{Individual Monotonicity (SH):}
If
$
 v(S) \leq v(S \cup i),  \forall S \subseteq M \backslash i,
$
then 
$
\text{SH}_i \geq 0.
$
\item \textbf{Pairwise Monotonicity (SH):}
If 
$
v(S \cup i) \leq v(S \cup j), \forall S \subseteq M \backslash \{i,j\},
$
then
$
\text{SH}_i \leq \text{SH}_j.
$
\item \textbf{Symmetric Monotonicity (SH):}
If
$
    w(S \cup i)-w(S)  \leq v(S \cup j) - v(S), \forall S \subseteq M \backslash \{i,j\},
$
then
$
        \text{SH}_i(w) \leq \text{SH}_j(v).
$
\end{itemize}
\end{definition}

Results for Shapley value for monotonicity have been studied by  \cite{casajus2018sign,chen2024attribution} with the following conclusion. 

 \begin{theorem}\label{thm:BShap_mono_global}
    BShap preserves Individual Monotonicity (SH), Pairwise Monotonicity (SH), and Symmetric Monotonicity (SH) only for $i = j$.
\end{theorem}





\section{Risk attribution methods and Shapley value}

 Previous efforts have primarily focused on BAMs, i.e. assigning attributions to a fixed explicand $\overline{\*x}$. In finance, risk is also critical. In the framework of Equation~\ref{eq:framework}, we want to figure out what the associated risks are for each random variable $X_i$. 
In light of this, we provide the following definition. 
\begin{definition}[\textbf{Risk Attribution Method (RAM)}]
    Given random variables $\*X, Y$, a function $f$, a baseline constant $\*x'$, and a risk measure $\varrho$, a risk attribution method (RAM) is any function of the form $\*A(\*X,\*x',Y,f,\varrho)$. For realizations of random variables $\{\*x_i\}_{i=1}^n$ and $\{y_i\}_{i=1}^n$, a baseline constant $\*x'$, a sample risk attribution method (SRAM) is any function of the form $\widehat{\*A}(\{\*x_i\}_{i=1}^n,\*x',\{y_i\}_{i=1}^n,\widehat{f},\widehat{\varrho})$, where $\widehat{f}$ is the learned function and $\widehat{\varrho}$ is the sample risk measure. We denote $A_i(\*X,\*x',Y,f,\varrho)$ and $\widehat{A}_i(\{\*x_i\}_{i=1}^n,\*x',\{y_i\}_{i=1}^n,\widehat{f},\widehat{\varrho})$ for the $i$th attribution of $\*A(\*X,\*x',Y,f,\varrho)$ and $\widehat{\*A}(\{\*x_i\}_{i=1}^n,\*x',\{y_i\}_{i=1}^n,\widehat{f},\widehat{\varrho})$. Often, we will use only $A_i(f)$ and $\widehat{A}_i(f)$ or $A_i$ and $\widehat{A}_i$ in short.
\end{definition}

\begin{remark}
    We distinguish between BAMs and RAMs by using  $\mathcal{A}_i$ and $A_i$ for the $i$-th attribution.
\end{remark}

\subsection{Shapley value for risk attribution methods}
We discuss how to extend the Shapley value for RAMs. In Shapley value's framework \eqref{eq:Shapley}, we define new characteristic functions 
\begin{align*}
    v_R(\*X,\*x',f,\varrho;S) = 
        \varrho(f(\*X_S; \*x'_{M \backslash S}))
\end{align*}
for RAMs and
\begin{align*}
    v_R(\{\*x_i\}_{i=1}^n,\*x',\widehat{f},\varrho;S) = \widehat{\varrho}(\{\widehat{f}((\*x_S)_i; (\*x'_{M \backslash S})_i)\}_{i=1}^n )
\end{align*}
for SRAMs. Often, we will use only $v_R(S)$ in short.  We follow BShap and use $\*x'$ as the baseline values. In practice, we pick $\*x'$ with constant to represent the case if the corresponding feature has no randomness. As an example, for a simple portfolio with two assets $f(X_1,X_2) = X_1+X_2$, let $\*x' = (0,0)$, then if $S = \{1\}$, we have $v_R(S) = \varrho(X_1)$. In this case, we ask what if $X_2$ is only a risk-free asset with zero return. We apply Shapley value~\eqref{eq:Shapley} with $v_R(S)$ and call the resulting method BShap (RAM) and denote the $i$-th attribution as $\text{BS}_i^R(\*X,\*x',Y,f,\varrho)$ or often $\text{BS}_i^R(f)$ and $\text{BS}_i^R$ in short.

\begin{example}\label{eg:eg_var}
    In general, no closed-form solutions are available. Consider $f(X_1,X_2) = X_1 + X_2$ with $\text{Var}(X_1) = \sigma_1^2$, $\text{Var}(X_2) = \sigma_2^2$, and $\text{corr}(X_1,X_2) = \rho$. For $\varrho(\cdot) = \text{STD}(\cdot)$ and $\*x' = (0,0)$, we calculate characteristic functions
    \begin{align*}
        & v_R(\{1\}) = \sigma_1, \ v_R(\{2\}) = \sigma_2, v_R(\{1,2\}) = \sqrt{\sigma_1^2 + \sigma_2^2 + 2\rho \sigma_1 \sigma_2}.
    \end{align*}
    As a result, we have
    \begin{align*}
        \text{BS}_1^R = \frac{\sigma_1}{2} + \frac{\sqrt{\sigma_1^2+\sigma_2^2 + 2\rho \sigma_1 \sigma_2}-\sigma_2}{2},  \text{BS}_2^R = \frac{\sigma_2}{2} + \frac{\sqrt{\sigma_1^2+\sigma_2^2 + 2\rho \sigma_1 \sigma_2}-\sigma_1}{2}.
    \end{align*}

\end{example}


\begin{remark}
We distinguish BShap between BAMs and RAMs by using $\text{BS}_i^B$ and $\text{BS}_i^R$ for the $i$-th attribution, respectively. Similarly, for characteristic functions, we use $v_R$ for RAMs, $v_B$ for BAMs, and $v$ for general Shapley value. \textbf{The Shapley value formula \eqref{eq:Shapley} applies the same to both problems and only the characteristic functions are different.} 
\end{remark}

\subsubsection{Portfolio theory}

In the case of portfolio theory, more results may be obtained. First, if we consider variance as the risk measure, there is a special case that allows us to obtain more analytical results, as discussed in  \cite{colini2018variance}.
    Consider a portfolio $X_1+\dots+X_m$, where $\sigma_i^2$ is the variance of $X_i$ and $\rho_{i,j}$ is the correlation between $X_i$ and $X_j$. Then BShap yields that
$
        \text{BS}_i^R = \sigma_i^2 +  \sum_{j \neq i}  \rho_{i,j} \sigma_i \sigma_j.
$

For more general cases, we consider the principle of diversification, which suggests that owning a variety of financial assets is less risky than owning only one. This is reflected in the following Sub-additivity property for risk measures \citep{artzner1999coherent,rockafellar2002conditional}. 


\begin{definition}[Sub-additivity]
    By the diversification principle, when two portfolios are combined, the risk cannot be greater than when they are combined separately. We say $\varrho$ preserves Sub-additivity if $Y_1, Y_2 \in \mathcal{L}$, 
$
    \varrho(Y_1+Y_2) \leq \varrho(Y_1) + \varrho(Y_2). 
$
\end{definition}

Therefore, we should expect to observe the impact of diversification when we calculate risk attributions. In the following proposition, we demonstrate how BShap preserves sub-additivity with proof in Appendix~\ref{sec:proof}. 

\begin{proposition} \label{prop:Sh_sub_addivitiy}
    For $f(\*X) = X_1 + \dots + X_m$, if $\varrho$ preserves Sub-additivity, then
    $\text{BS}^R_i \leq \varrho(X_i)$. As a result of sub-additivity, BShap (RAM) could reflect the principle of diversification, which reduces each individual risk attribution when assets are added to the portfolio. 
\end{proposition}

\subsection{Fundamental axioms}

We present fundamental axioms for RAMs that are analogous to Section~\ref{sec:BAM} and are expected to hold in general.  A similar idea can be applied to SRAMs. \textbf{We mainly use RAMs for demonstration of analytical examples and SRAMs in practice.} 

\begin{definition}
\textbf{Fundamental axioms for RAMs}. 

$\bullet$ \textbf{Completeness (RAM):} 
    $\sum_{i=1}^m A_i = \varrho(f(\*X)) - \varrho(f(\*X'))$.

$\bullet$ \textbf{Dummy (RAM):}
    If $\frac{\partial f(\*x)}{\partial x_i} = 0$, then $A_i = 0$.

$\bullet$ \textbf{Symmetry (RAM):}
    Define $\*x^*$ by swapping $x_i$ and $x_j$ and the probability density function of $f(\*X)$ by $g(\*x)$.  
    Suppose $g(\*x) = g(\*x^*)$, $\forall \*x$. Then if $x_i' = x_j'$, then $A_i = A_j$.
\end{definition}

\begin{theorem} \label{thm:BShap_RAM_basic_axioms}
    BShap (RAM) preserves Completeness (RAM), Dummy (RAM), and Symmetry (RAM). 
\end{theorem}

\section{Comparisons of BAMs and RAMs}

RAMs and BAMs differ primarily in that the former focuses on the mean attributions whereas the latter focuses on the risk attributions. Consider the following portfolio $f(X_1,X_2) = X_1 + X_2$, in which $X_1$ follows a normal distribution with $\mathcal{N}(\mu,\sigma^2)$ and $X_2$ is a constant $r$, which can be considered coming from the risk-free interest rate. Assume that the standard deviation is used to measure risk. By RAMs, we expect that there is no risk associated with $X_2$ and therefore $A_2=0$. By BAMs, there is an attribution from $X_2$ with $\mathcal{A}_2 = r$. The focus of BAMs and RAMs, therefore, differs significantly. Following this, we illustrate some differences between BAMs and RAMs in terms of axioms.

\subsection{Linearity}

As for BAMs, the Linearity (BAMs) implies the Linearity (SH). This is no longer true for RAMs due to the nonlinearity of risk measures. Consider an example 
$        f(X_1,X_2) = X_1, \ g(X_1,X_2) = X_2
$
    where the variance of $X_1$ and $X_2$ are $\sigma_1^2$ and $\sigma_2^2$ and the correlation between them is $\rho$. Suppose $\varrho(\cdot) = \text{STD}(\cdot)$. Then by Completeness (RAM) and Dummy (RAM), we have
    \begin{align*}
        A_1(f) = \sigma_1, A_2(f) = 0, A_1(g) = 0, A_2(g) = \sigma_2.
    \end{align*}
    If Linearity (BAM) holds, then we have
$
        A_i(f+g) = A_i(f) + A_i(g).
$
    However, if this were true,
    \begin{align*}
        A_1(f+g) + A_2(f+g) 
        = A_1(f) + A_1(g) + A_2(f) + A_2(g) = \sigma_1 + \sigma_2 \neq \sqrt{\sigma_1^2 + \sigma_2^2 + 2 \rho \sigma_1 \sigma_2} = \varrho(f+g),
    \end{align*}
    which is in contrast to the Completeness (RAM). \textbf{Therefore, Linearity is incompatible with Completeness (RAM) and Dummy (RAM).} Incompatible means that these axioms cannot be held simultaneously. \textbf{In view of this, there is no corresponding Linearity axiom for RAMs as there is for BAMs.}

While Linearity is not required for RAMs, BShap (RAM) does require Linearity (SH) for characteristic functions. It should be noted that this may not be a strict requirement for RAMs. As for Example~\ref{eg:eg_var}, by Euler's decomposition \cite{meucci2005risk}, we can decompose the STD of a portfolio
\begin{align*}
    \text{STD}(f(X_1,X_2)) = \frac{ \sigma_1^2 +  \rho \sigma_1 \sigma_2}{\sqrt{\sigma_1^2+\sigma_2^2+2\rho \sigma_1 \sigma_2}} 
    + \frac{ \sigma_2^2 +  \rho \sigma_1 \sigma_2}{\sqrt{\sigma_1^2+\sigma_2^2+2\rho \sigma_1 \sigma_2}}.
\end{align*}
It can be easily observed that Completeness (RAM), Dummy (RAM), and Symmetry (RAM) are preserved by this formula. Consequently, Euler's decomposition also provides a fair allocation of risk. 
In terms of characteristic functions, it can be rewritten as
\begin{align*}
    \text{STD}(f(X_1,X_2)) = \frac{v_R^2(\{1,2\})+v_R^2(\{1\})-v_R^2(\{2\})}{2v_R(\{1,2\})} 
    + \frac{v_R^2(\{1,2\})+v_R^2(\{2\})-v_R^2(\{1\})}{2v_R(\{1,2\})}.
\end{align*}
It is clear that this formula is not linear in the sense of Linearity (SH). Although it's still easy to interpret a linear function by Euler's decomposition, for a general complex $f$, a nonlinear interpretation might be too complicated. \textbf{Therefore, Linearity (SH) for BShap mainly helps make the explanation formula more transparent in the context of RAMs.} 

\begin{remark}
    Euler's decomposition is only applicable to homogeneous functions of degree one. Thus, it cannot be directly generalized to a general complex nonlinear function $f$. 
\end{remark}

\subsection{Symmetry}

The concept of symmetry is generally concerned with the fairness of allocation. Two features that behave the same should receive the same attribution. Different from BAM, symmetry (RAM) further requires that $X_i$ and $X_j$ are also symmetric in $\*X$ since the distribution of features matters. By directly applying BShap (BAM), the symmetry of features may be overlooked, resulting in unfair allocations, as illustrated in the following example. 

\begin{example}
    
    As a simple example, consider a portfolio $Y_t = f(X_{1,t},X_{2,t}) = X_{1,t} + X_{2,t}$ as the daily return of the portfolio with two assets. Suppose we use the Value-at-Risk (VaR) as the risk measure that $\widehat{\varrho}(\cdot) = \widehat{\text{VaR}}_{\alpha}(\cdot)$, and we find $y_{\tau} = -\widehat{\text{VaR}}_{\alpha} (Y_t)$ at date $\tau$. Suppose we apply BShap (BAM) to $y_{\tau}$ and have $\text{BS}_i(v_B) = x_{i,\tau}$. This calculation focus only on $x_{1,\tau}$ and $x_{2,\tau}$ and neglects the distribution of $X_{1,t}$ and $X_{2,t}$. Suppose $\{x_{1,t}+x_{2,t}\}_{i=1}^n$ is symmetric about $\{x_{1,t}\}_{i=1}^n$ and $\{x_{2,t}\}_{i=1}^n$. By Symmetry (RAM), we would expect to have $A_1 = A_2$. It may be the case, however, that $x_{1,\tau}$ contributes most to $y_{\tau}$ such that $\text{BS}_1(v_B)$ is much greater than $\text{BS}_2(v_B)$ when $x_{1,\tau}$ is much less than $x_{2,\tau}$. Thus, the Symmetry (RAM) is violated and the explanation is unfair. 
\end{example}

\subsection{Monotonicity axioms}

There have been extensive studies of domain knowledge-inspired monotonicity on attribution methods, \citep{friedman1999three,lundstrom2023four,chen2023address,chen2022monotonic,young1985producer}. The monotonicity axiom enables us to ensure the sign of attributions and to compare the relative importance of features. It is often possible to infer monotonicity from the derivatives of functions for BAMs. Therefore, the following monotonicity axioms for BAMs are often considered \citep{lundstrom2023four}.

\begin{definition}\label{def:mono_axiom_BAM} 
\textbf{Monotonic axioms for BAMs}. 
\begin{itemize}
\item \textbf{Individual Monotonicity (BAM):}
    Suppose $\frac{\partial f(\*x)}{\partial x_{i}} \geq 0$, then
$
        \mathcal{A}_{i} \geq 0.
$
\item \textbf{Pairwise Monotonicity (BAM):}
If $ \frac{\partial f(\*x)}{\partial x_i} \leq \frac{\partial f(\*x)}{\partial x_j}$ and $x_i \leq x_j$, then 
$
\mathcal{A}_i \leq \mathcal{A}_j.
$
\item \textbf{Symmetric Monotonicity (BAM):}
    If $\frac{\partial f(\*x)}{\partial x_i} \leq \frac{\partial g(\*x)}{\partial x_j}$ and $x_i \leq x_j$, then
$
\mathcal{A}_i(f) \leq \mathcal{A}_j(g).
$
\end{itemize}
\end{definition}        

\begin{example} \label{eg:mono_BAMs}
We provide the following examples in credit scoring to demonstrate monotonicity.
\begin{itemize}
\item The probability of default increases with the number of past-due payments. Due to this, the corresponding attribution should also be positive. 
\item Suppose we use $x_i$ and $x_j$ to count the number of past dues less than and more than three months, then for each additional past due, $x_j$ should have a greater impact on the probability of default. Accordingly, if $x_j \geq x_i$, we should expect a larger attribution for $x_j$. 
\item Suppose that we compare two different credit scoring models for $x_i$ and $x_j$, we expect to compare their attributions if there exists clear monotonicity. 
\end{itemize}
\end{example}

As demonstrated in the following example, we can also expect some monotonicity for RAMs. 

\begin{example}
Consider nonlinear factor models 
$
    Y_t^1 = f(X_{1,t},X_{2,t}), Y_t^2 = g(X_{1,t},X_{2,t}),
$
where at $t$-th day, $Y_t^1$ and $Y_t^2$ are daily returns of Google and Capital One, and $X_{1,t}$ and $X_{2,t}$ are daily returns of the entire Finance sector and Technology sector. 
\begin{itemize}
\item Google is expected to have positive exposure to risk of the Technology sector. 
\item Google is more exposed to the risk of the Technology factor than to the Financial factor.
\item Google has a larger exposure to the risk of the Technology sector than Capital One. 
\end{itemize}
\end{example}

For BAMs, monotonicity is studied from the standpoint of the derivative of the function, e.g., $\frac{\partial f}{\partial x_i}(\*x) \geq 0$. This condition, however, may not be directly applicable to RAM content. 
Consider a simple example of $f(X_1) = X_1$,
where $X_1$ follows the standard normal distribution. While $f$ and $-f$ has the opposite signs, the distributions of $f(X_1)$ and $-f(X_1)$ are the same and we have $\varrho(f) = \varrho(-f)$. Therefore, Monotonicity on the function $f$ is irrelevant. For RAMs, we instead focus on characteristic functions as in Section~\ref{sec:Shap}. However, as demonstrated below, in contrast to BAMs \citep{lundstrom2023four}, Symmetryic Monotonicity (SH) for RAMs when $i \neq j$ is not expected. 


\begin{example} \label{eg:SM_incompatible}
Symmetric Monotonicity (RAM) for $i \neq j$ is incompatible with Dummy (RAM) and Completeness (RAM). 
Consider 
\begin{align*}
    f(X_1,X_2) = X_1 + X_2, \ g(X_1,X_2) = X_2, \ h(X_1,X_2) = X_1. 
\end{align*}
 By Dummy (RAM) and Completeness (RAM), we have $A_2(g) = \varrho(X_2)$ and $A_1(h) = \varrho(X_1)$. By Symmetric Monotonicity (RAM) for $i \neq j$, we  have $A_1(f)=A_2(g)$ and $A_2(f) = A_1(h)$. Then we no longer have Completeness (RAM). The relative importance of risk attributions for different features in different models cannot, therefore, be directly compared. 

 For a specific example of BShap (RAM), suppose $\text{Var}(X_1) = \text{Var}(X_2)= \sigma^2$ and $\rho>0$. By Symmetry (RAM), we know that to $f$, $X_1$ and $X_2$ bring the same amount of risk. Furthermore, by Theorem~\ref{thm:BShap_mono_global}, $X_2$ brings more risk to $f$ than to $g$. However, we cannot conclude that $X_1$ brings the same amount of risk to $f$ as the $X_2$ to $g$ even though they have the same distribution. By BShap (RAM),  we have
    \begin{align*}
        \textbf{BS}^R(f) = \left[ \begin{matrix} \sigma^2 + \rho \sigma^2 \\
        \sigma^2 + \rho \sigma^2 \end{matrix} \right], \
        \textbf{BS}^R(g) = \left[ \begin{matrix} 0 \\ \sigma^2 \end{matrix} \right].
    \end{align*}

\end{example}

\section{Empirical examples}
In this section, we demonstrate how BShap allocates risk attributions using a variety of examples. To the best of our knowledge, no other methods have been applied to allocate risk for general models in the past. 

\subsection{Portfolio analysis (linear case)}

We include daily log returns of Microsoft, Pfizer, Tesla, JP Morgan, and Netflix from 2011 to 2023. 
WLOG, we denote them as $X_1-X_5$. These stocks are chosen from different industry sectors.
In practice, portfolio optimization has been extensively used to diversify risk \citep{rockafellar2000optimization,rockafellar2002conditional}. For this example, we consider minimizing Conditional Value-at-Risk (CVaR) using $\alpha=0.05$ for the selected five stocks in 2023,
\begin{align*}
\begin{cases}
   & \min_{\*c} \text{CVaR}_{\alpha}(\*X \*c), \\
   & \sum_{i} c_i = 1 \text{ and } c_i \geq 0, \forall i.
\end{cases}
\end{align*}
By optimization, we obtain that
$
\*c = \left[ \begin{matrix}
        0.42 & 0.19 & 0.0033 & 0.30 & 0.083 \end{matrix} \right].
$
The original CVaR for returns ($\text{CVaR}_{\alpha}(X_i)$ for the $i$-th entry) are 
\begin{align*}
    \left[ \begin{matrix} 0.030 & 0.038 & 0.069 & 0.031 & 0.043 \end{matrix} \right].
\end{align*}
In the optimized portfolio, the $\text{CVaR}_{\alpha}(\*X \*c)$ is reduced to $0.019$, which demonstrates the effectiveness of the optimization. Applying BShap (RAM), with the help of Proposition~\ref{prop:Sh_sub_addivitiy} for Sub-additivity, we are able to see the reduction of risks for each component by comparing the CVaR of $X_i c_i$ in Figure~\ref{fig:port_CVaR_original} and the risk attributions of the portfolio in Figure~\ref{fig:port_CVaR_opti}. Among the portfolios, Microsoft's CVaR has been reduced the most, which makes sense since Microsoft has been selected with the greatest exposure. 


\begin{figure}[ht]
\centering
\begin{minipage}{0.5\textwidth}
  \includegraphics[scale=0.35]{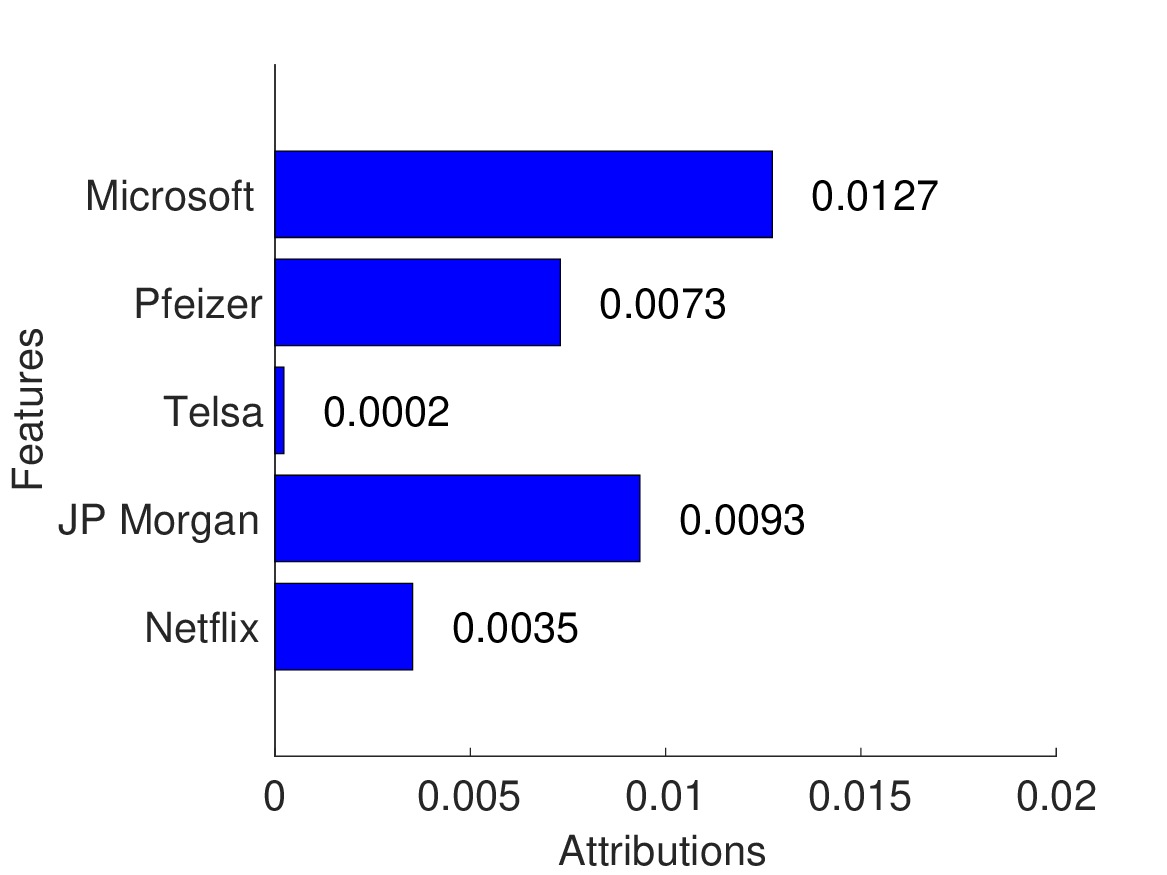}
    \caption{CVaR of each stock with $c_i$}
    \label{fig:port_CVaR_original}
\end{minipage}%
\begin{minipage}{0.5\textwidth}
  \includegraphics[scale=0.35]{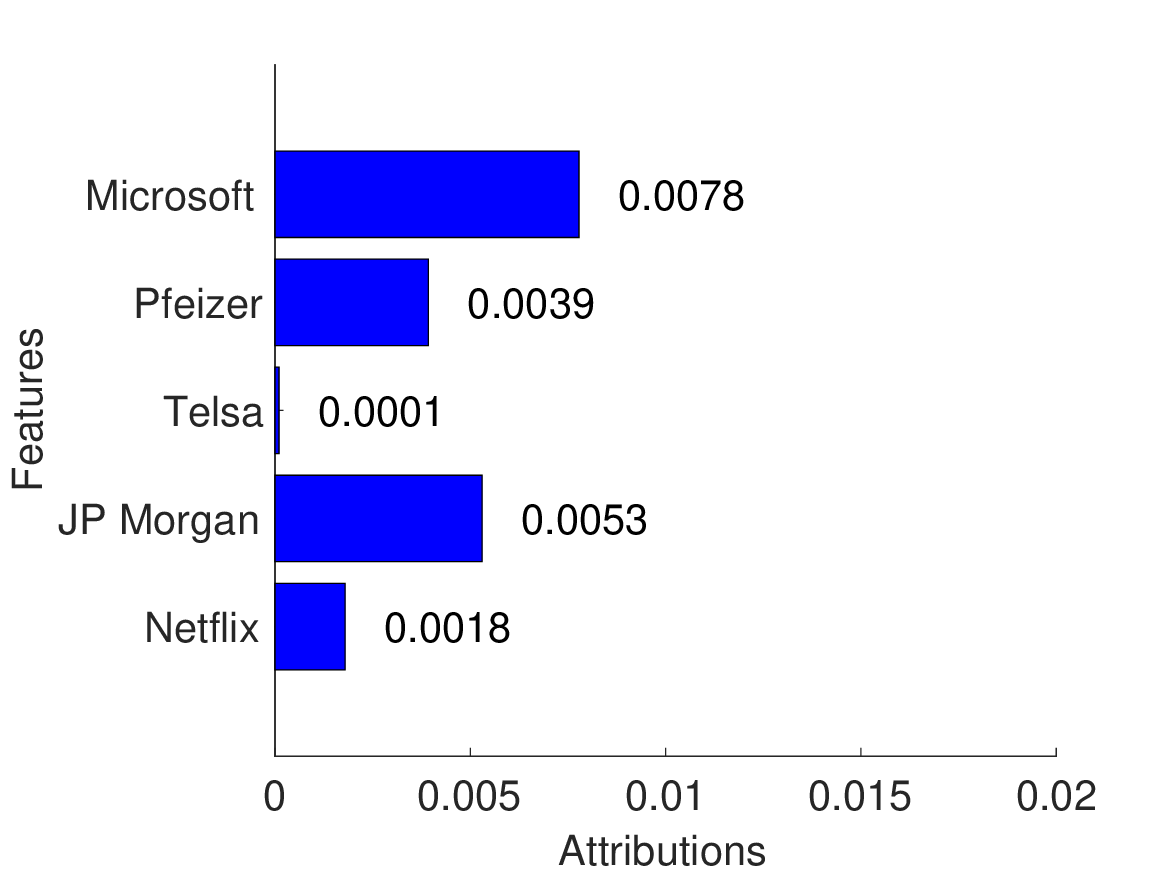}
    \caption{Risk attributions (CVaR) of the portfolio}
    \label{fig:port_CVaR_opti}
\end{minipage}
\end{figure}

\subsection{Nonlinear factor models}

We apply a neural network to the daily log return of Tesla and Google (Alphabet) from 2011 to 2023 with industry factors as features. The stock prices of Google and Tesla were collected from Yahoo Finance. We downloaded industry factors from the French's website\footnote{\url{https://mba.tuck.dartmouth.edu/pages/faculty/ken.french/data_library.html}}. A total of 12 industry factors are used. In order to make it easier to understand, we have renamed Business Equipment as Technology and Money as Finance. 

For neural networks, we use the architecture of $[16,8]$ with Relu activations and $l2$ regularization. Using the conjugate gradient, we solve the optimization problem and stop iterating after 400 steps. A training set consists of data from 2011 to 2018, a validation set consists of data from 2019 to 2020, and a test set consists of data from 2021 to 2023. Using the validation set, the regularization parameters are determined through grid searching. The squared root of the mean squared error is used to measure error. For Tesla, the error rate is 0.012 and for Google, it is 0.011. The performance of the model can certainly be improved, but that is not the primary objective of this study. The model is only used for demonstration purposes. 


\paragraph{SRAM for factor models}
For factor models, it is important to note that the residual term $\epsilon$ also contributes to the risk. In this example, $\epsilon$ represents the idiosyncratic risk, which is the risk specific to the company. Consequently, $\epsilon$ must be added as an additional feature, and we denote it as the $m+1$ feature. Consequently, we consider the function
$
    \widetilde{f}(\*x,\widehat{\epsilon}) = \widehat{f}(\*x) + \widehat{\epsilon},
$
whereas $\widehat{\epsilon}_i = y_i - \widehat{f}(\*x_i)$ for SRAM. Then BShap (RAM) is applied to $\widetilde{f}$. For this example, we consider $\varrho(\cdot) = \text{CVaR}_{5\%}(\cdot)$.

\paragraph{Results}

Risk attributions ($\times 100$) for Tesla and Google in 2023 are plotted in Figure~\ref{fig:TSLA_CVaR} and \ref{fig:GOOG_CVaR}.
Results are in accordance with domain knowledge and Theorem~\ref{thm:BShap_mono_global}.

$\bullet$ BShap (RAM) produces positive risk attributions for relevant risky factors. 

$\bullet$ BShap (RAM) has ranked features by their risk level. Tesla, for example, has the largest exposure to Consumer Durables since it focuses on electric vehicles, followed by Manufacturing, Technology, and Finance. 

$\bullet$ BShap (RAM) has been able to compare the risk levels of different assets for the same factor. For example, in line with expectations, Google is much more exposed to Technology than Tesla. 

$\bullet$ Idiosyncratic risks are important. As an example, it is the largest risk factor for Google this year. It is consistent with financial analysts' perception that a large portion of Google's risk this year is due to Google's specific risk. For example, Google's stock price fell by more than $9\%$ in October 2023 due to disappointing cloud computing earnings\footnote{https://www.wsj.com/livecoverage/stock-market-today-dow-jones-10-25-2023/card/alphabet-s-stock-falls-as-cloud-slowdown-outweighs-strong-earnings-59VFE9hO7dst4QDfGFBg}.



\begin{figure}[ht]
\centering
\begin{minipage}{0.5\textwidth}
  \includegraphics[scale=0.35]{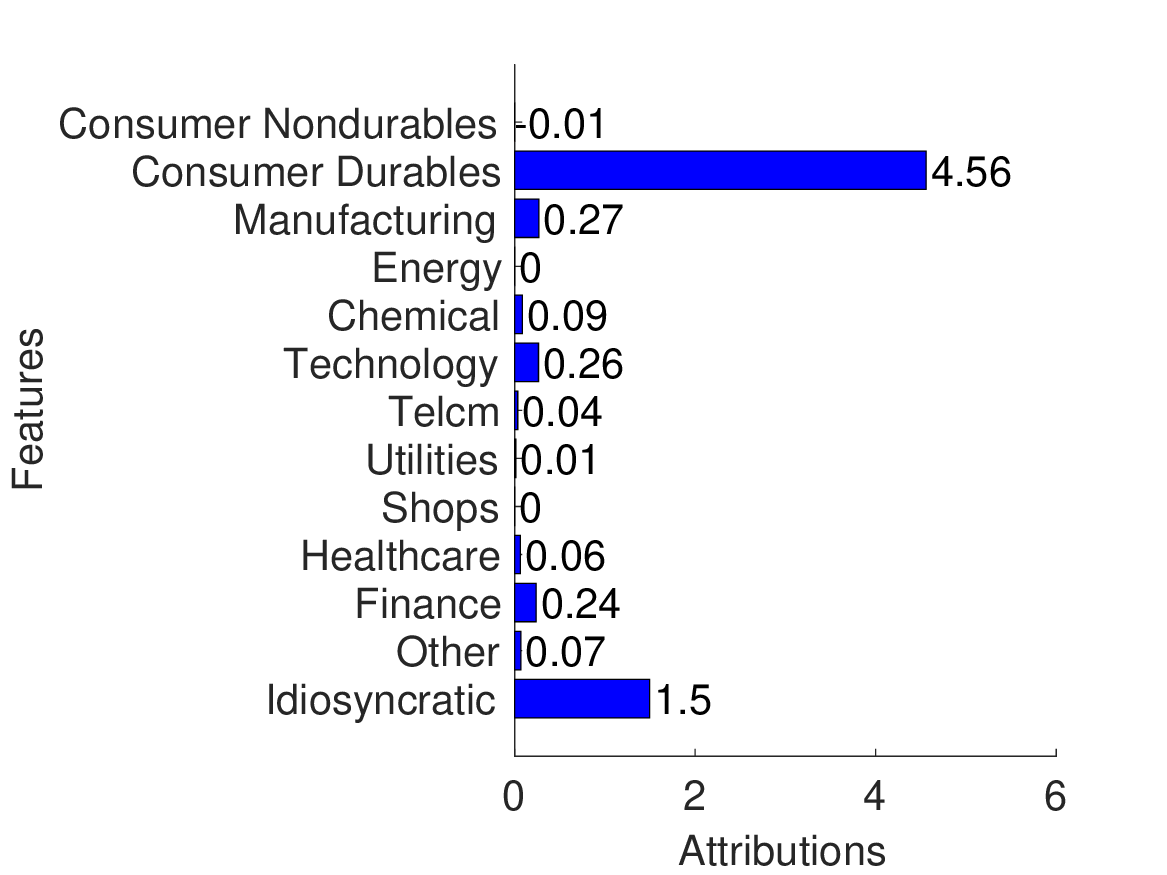}
    \caption{BShap ($\text{CVaR}_{5\%}$) of Tesla in 2023}
    \label{fig:TSLA_CVaR}
\end{minipage}%
\begin{minipage}{0.5\textwidth}
  \includegraphics[scale=0.35]{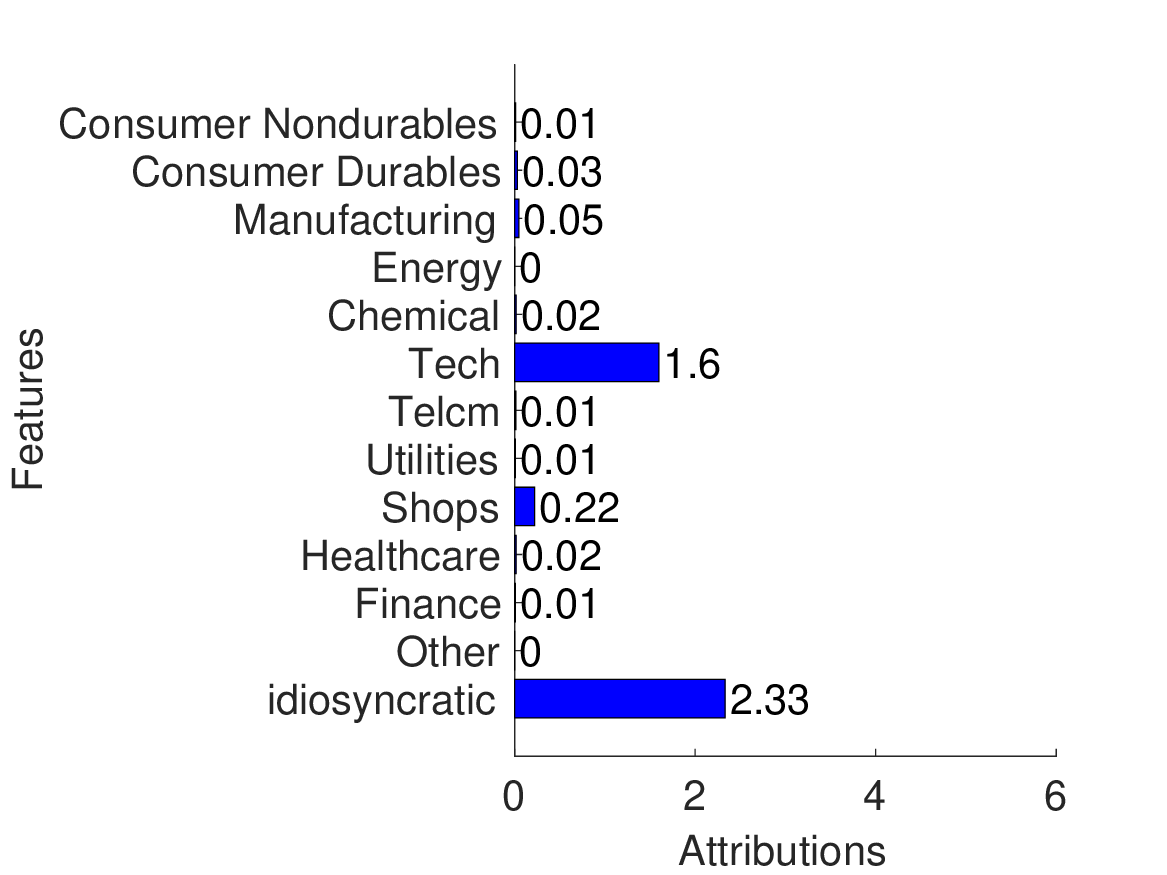}
    \caption{BShap ($\text{CVaR}_{5\%}$) of Google in 2023}
    \label{fig:GOOG_CVaR}
\end{minipage}
\end{figure}



\subsection{Option pricing}

We provide an example of option pricing. A call option is a contract under which a buyer has the right, but not the obligation, to purchase an underlying asset at a strike price of $K$ with a maturity period of $\tau$. Let us assume that the present time is $t$, and that, at time $t+\tau$, a stock price greater than the strike price will exercise the option and the payoff will be $S_{t+\tau}-K$; otherwise, the option will not be exercised and the payoff will be 0. In summary, the value of the call option at time $t+\tau$ is equal to $C(S_{t+\tau},K,\tau) = (S_{t+\tau}-K)^+$. Based on a couple of assumptions, \cite{black1973pricing,merton1973theory} developed a pricing formula $C(S_t,K,\tau,\sigma,r)$, where $\sigma$ represents volatility and $r$ represents the constant risk-free interest rate. In this demonstration, the Black-Scholes-Merton (BSM) formula is used for the call option price for demonstration since it is widely accepted in industries. However, ML models may be used for potential improvements. 

We collected stock prices and option data from Wharton Research Data Services in 2008. There are 253 trading days. The London Interbank Offered Rate (LIBOR) is used to represent risk-free interest rates. The LIBOR served as the benchmark interest rate at which major global banks lent to one another in the international interbank market for short-term loans. This key benchmark interest rate served as an indication of borrowing costs between banks throughout the world. Daily volatility is measured by the VIX index \citep{exchange2009cboe}, which is calculated by option data.

Suppose we wish to know the risk associated with the purchase of a call option. We are particularly concerned about whether the option price will change significantly tomorrow. The S$\&$P 500 index data for 2008 is used. As of the end date ($t=T$) of 2008, the stock price $S_T$ is approximately 890. At the time $T$, we are interested in a call option with a strike price of $K=800$ and an expiration length of one month. We need to evaluate how the option prices will change on the next date. Taking a look back at the year 2008, we can see how the market conditions changed in just one day. We use $\delta_i = \frac{S_{i+1}}{S_i}$ to denote the percentage change in the stock, $\sigma_i$ to denote volatility, and $r_i$ to denote risk-free interest rates at $i$th day. Since the stock price is not stationary over time, we consider $S_T \delta_i$ as the change of the potential stock price in one day for each date. Thus, we consider all possible scenarios in the last year using $\{ C(S_T \delta_i, K=800, T = \frac{30}{365}, \sigma _{i+1}, r_{i+1}) \}_{i=1}^{n-1}$. In other words, we consider the possibility that the conditions of the market tomorrow will be similar to those of $i+1$ day, which would result in a change in stock prices of $\delta_i-1$ percent, a change in volatility of $\sigma_{i+1}$, and a change in the risk-free interest rate of $r_{i+1}$. We use $\widehat{\varrho}(\cdot) = \widehat{\text{STD}}(\cdot)$ as an example of a demonstration. In this period, $\widehat{\text{STD}}(C(S_{T} \delta_i, K=800, T = \frac{30}{365}, \sigma_{i+1}, r_{i+1})) = 19.6$, which indicates that the price of the option fluctuates greatly. BShap (RAM) is then used to decompose the risk. To simplify interpretation, we use the log of stock prices, volatility, and interest rates as inputs, since otherwise they are not on the same scale. There is a natural baseline point here, which would be $\*x'=(S_T,\sigma_T,r_T)$, which is the current market condition. When the baseline value $\*x'$ is a constant, $\widehat{\text{STD}}(\*x') = 0$. According to BShap (RAM), log stock price, log volatility, and log interest rate have risk attributions of 13.8, 5.7, and 0.1, respectively. The results are in accordance with our expectations. A European call option is likely to be affected most by the stock price, followed by volatility, and the interest rate has the least impact since it is relatively stable in comparison to the other factors. Based on risk attributions, one may also want to consider hedging the stock price against the risks. Moreover, one may hedge the risks associated with volatility.

\section{Conclusion} \label{sec:conclusion}

This study examines risk attribution methods by extending the Shapley value framework. We demonstrate that attributions of risk to general complex models, including machine learning models, can be fairly allocated through analysis and empirical examples. This has resulted in a better understanding of risk attributions in practice. 

\section*{Disclosure of interest}
The authors report that there are no competing interests to declare.

\section*{Funding}
No funding was received.

\bibliography{ref.bib}
\bibliographystyle{plainnat}

\appendices 

\section{Proof}\label{sec:proof}


    




\begin{proof}[Proof of Theorem~\ref{thm:BShap_RAM_basic_axioms}]
    The Completeness (RAM) directly follows from the Completeness (SH) with the choice of characteristic function $v_R$. Under assumptions of Dummy (RAM) and Symmetry (RAM), the assumptions of the characteristic functions for Dummy (SH) and Symmetry (SH) are satisfied, and therefore the conclusions hold. 
\end{proof}

\begin{proof}[Proof of Proposition~\ref{prop:Sh_sub_addivitiy}]
    We calculate that 
    \begin{align*}
        \text{BS}_i &= \sum_{S \subseteq M \backslash i} \frac{|S|! (|M|-|S|-1)!}{|M|!} (v_R(S \cup i) - v_R(S)) \\
        &= \sum_{S \subseteq M \backslash i} \frac{|S|! (|M|-|S|-1)!}{|M|!} \left( \varrho \left( \sum_{j \in S} X_j + X_i \right) - \varrho \left( \sum_{j \in S} X_j \right) \right) \\
        &\leq \sum_{S \subseteq M \backslash i} \frac{|S|! (|M|-|S|-1)!}{|M|!} \varrho(X_i) = \varrho(X_i).
    \end{align*}

\end{proof}

\end{document}